\documentclass[12pt]{article}    

\usepackage[margin=0.75in]{geometry} 
\usepackage{amsmath,amssymb,amsthm}

\newtheorem{theorem}{Theorem}[section]
\newtheorem{proposition}[theorem]{Proposition}

\newcommand{\be}{\begin{equation}}
\newcommand{\ee}{\end{equation}}

\newcommand{\bea}{\begin{eqnarray}}
\newcommand{\eea}{\end{eqnarray}}

\numberwithin{equation}{section}
\begin{document}

\title{Painlev\'{e} V, Painlev\'{e} XXXIV and the Degenerate Laguerre Unitary Ensemble}
\author{Chao Min\thanks{School of Mathematical Sciences, Huaqiao University, Quanzhou 362021, China; e-mail: chaomin@hqu.edu.cn} and Yang Chen\thanks{Correspondence to: Yang Chen, Department of Mathematics, Faculty of Science and Technology, University of Macau, Macau, China; e-mail: yangbrookchen@yahoo.co.uk}}


\date{\today}
\maketitle
\begin{abstract}
In this paper, we study the Hankel determinant associated with the degenerate Laguerre unitary ensemble. This problem originates from the largest or smallest eigenvalue distribution of the degenerate Laguerre unitary ensemble. We derive the ladder operators and its compatibility condition with respect to a general perturbed weight. By applying the ladder operators to our problem, we obtain two auxiliary quantities $R_n(t)$ and $r_n(t)$ and show that they satisfy the coupled Riccati equations, from which we find that $R_n(t)$ satisfies the Painlev\'{e} V equation. Furthermore, we prove that $\sigma_{n}(t)$, a quantity related to the logarithmic derivative of the Hankel determinant, satisfies both the continuous and discrete Jimbo-Miwa-Okamoto $\sigma$-form of the Painlev\'{e} V. In the end, by using Dyson's Coulomb fluid approach, we consider the large $n$ asymptotic behavior of our problem at the soft edge, which gives rise to the Painlev\'{e} XXXIV equation.
\end{abstract}

$\mathbf{Keywords}$: Hankel determinant; Degenerate Laguerre unitary ensemble; Ladder operators;

Orthogonal polynomials; Painlev\'{e} equations; Asymptotics.

$\mathbf{Mathematics\:\: Subject\:\: Classification\:\: 2010}$: 15B52, 42C05, 33E17.

\section{Introduction}
In random matrix theory, it is well known that the partition function of a unitary ensemble is given by \cite{Mehta}
$$
\Delta_{N}[w_{0}]:=\frac{1}{N!}\int_{[a,b]^N}\prod_{1\leq i<j\leq N}(x_{i}-x_{j})^2\prod_{k=1}^{N}w_{0}(x_{k})dx_{k},
$$
where $\{x_{j}\}_{j=1}^{N}$ are the eigenvalues of $N\times N$ Hermitian matrices from the unitary ensemble, and $w_{0}(x)$ is a weight function supported on an interval $[a,b]$.

For the generic Laguerre unitary ensemble, $w_{0}(x)=x^{\alpha}\mathrm{e}^{-x},\;\alpha>-1,\;x\in [0,\infty)$.
In the case of a single degenerate eigenvalue $t$ with $K$ fold degeneracy and the rest $n$ eigenvalues, which are also denoted by $\{x_{j}\}_{j=1}^{n}$ for simplicity, are distinct, such that $N=n+K$, we find the partition function reads,
$$
\Delta_{n+K}[w_{0}]=\int_{0}^{\infty}t^{K\alpha}\mathrm{e}^{-Kt}D_{n}(t)dt,
$$
where
$$
D_{n}(t):=\frac{1}{n!}\int_{[0,\infty)^n}\prod_{1\leq i<j\leq n}(x_{i}-x_{j})^2\prod_{l=1}^{n}(x_{l}-t)^{2K}x_{l}^{\alpha}\mathrm{e}^{-x_{l}}dx_{l}.
$$
It is well known that $D_{n}(t)$ can be expressed as the following Hankel determinant \cite{Mehta},
$$
D_{n}(t)=\det\left(\int_{0}^{\infty}x^{i+j}x^{\alpha}\mathrm{e}^{-x}|x-t|^{2K}dx\right)_{i,j=0}^{n-1}.
$$
We mention that Chen and Feigin \cite{ChenFeigin} studied this kind of degenerate unitary ensemble but for the Gaussian case.

More generally, we consider the Hankel determinant generated by the perturbed Laguerre weight, namely,
\be\label{dnt}
\mathcal{D}_{n}(t):=\det\left(\int_{0}^{\infty}x^{i+j}w(x,t)dx\right)_{i,j=0}^{n-1},
\ee
where
$$
w(x,t):=x^{\alpha}\mathrm{e}^{-x}|x-t|^{\gamma}(A+B\theta(x-t)),\;\; x\geq 0,\;t\geq 0,\;\alpha>0,\; \gamma>0.
$$
Here $\theta(x)$ is the Heaviside step function, i.e., $\theta(x)$ is 1 for $x>0$ and 0 otherwise; $A$ and $B$ are constants and $A\geq 0,\; A+B\geq 0$.

We would like to point out some special cases of our problem. If $\gamma=2K,\;A=1,\;B=0$, the Hankel determinant $\mathcal{D}_{n}(t)$ is related to the partition function of the degenerate Laguerre unitary ensemble (dLUE); if $\gamma=2K,\;A=0,\;B=1$, it allows us to compute the probability that all the eigenvalues are not less than $t$ in the dLUE; if $\gamma=2K,\;A=1,\;B=-1$, it corresponds to the probability that all the eigenvalues are not greater than $t$ in the dLUE. Furthermore, the $\gamma=0$ case has been studied by Basor and Chen \cite{Basor2009}. They also investigated the Hankel determinant generated by the weight $\tilde{w}(x,t)=x^{\alpha}\mathrm{e}^{-x}(x+t)^{\gamma},\;x>0,\;t>0$, which is related to the information theory of MIMO wireless systems \cite{Basor2015}. Note that the problem on the weight $\tilde{w}(x,t)$ is different from ours, since our weight vanishes at a singular point $t$ in the interior of the support. Finally, we mention that for $\gamma>-1$, the weight $w(x,t)$ is called the Laguerre weight with a Fisher-Hartwig singularity \cite{Fisher}. Rencently, Wu, Xu and Zhao \cite{Wu} studied the Hankel determinant for the Gaussian weight perturbed by a Fisher-Hartwig singularity at the soft edge.

We now introduce some elementary facts about the orthogonal polynomials. Let $P_{n}(x,t)$ be the monic polynomials of degree $n$ orthogonal with respect to the weight $w(x,t)$,
\be\label{ops}
\int_{0}^{\infty}P_{m}(x,t)P_{n}(x,t)w(x,t)dx=h_{n}(t)\delta_{mn},\;\;m, n=0,1,2,\ldots.
\ee
We write $P_{n}(x,t)$ in the following expansion form,
\be\label{expan}
P_{n}(x,t)=x^{n}+\mathrm{p}(n,t)x^{n-1}+\cdots,
\ee
and we will see that $\mathrm{p}(n,t)$, the coefficient of $x^{n-1}$, plays a significant role in the following discussions.

For the orthogonal polynomials $P_{n}(x,t)$, we have the three-term recurrence relation \cite{Szego,Chihara}
\be\label{rr}
xP_{n}(x,t)=P_{n+1}(x,t)+\alpha_{n}(t)P_{n}(x,t)+\beta_{n}(t)P_{n-1}(x,t)
\ee
with the initial conditions
$$
P_{0}(x,t)=1,\;\;\beta_{0}(t)P_{-1}(x,t)=0.
$$
An easy consequence of (\ref{ops}), (\ref{expan}) and (\ref{rr}) gives
\be\label{al}
\alpha_{n}(t)=\mathrm{p}(n,t)-\mathrm{p}(n+1,t)
\ee
and
\be\label{be}
\beta_{n}(t)=\frac{h_{n}(t)}{h_{n-1}(t)}.
\ee
A telescopic sum of (\ref{al}) yields
\be\label{sum}
\sum_{j=0}^{n-1}\alpha_{j}(t)=-\mathrm{p}(n,t).
\ee
Finally, it is well known that \cite{Ismail}
\be\label{hankel}
\mathcal{D}_{n}(t)=\prod_{j=0}^{n-1}h_{j}(t).
\ee

The rest of this paper is organized as follows. In Sec. 2, we derive the ladder operators and the compatibility conditions with respect to the weight
$w(x):=w_{0}(x)|x-t|^{\gamma}(A+B\theta(x-t))$, where $w_{0}(x)$ is a general smooth weight. In Sec. 3, we apply the ladder operators for the general case to the perturbed Laguerre weight and obtain some important identities on the auxiliary quantities $R_n(t)$ and $r_n(t)$. In Sec. 4, we show that $R_n(t)$ and $r_n(t)$ satisfy the coupled Riccati equations, which give rise to the Painlev\'{e} V equation satisfied by $R_n(t)$. We also prove that a quantity $\sigma_{n}(t)$, allied to the logarithmic derivative of the Hankel determinant, satisfies both the continuous and discrete $\sigma$-form of the Painlev\'{e} V.
In Sec. 5, we consider the large $n$ asymptotics of our problem at the edge, from which the Painlev\'{e} XXXIV equation appears.

\section{Ladder Operators and Compatibility Conditions}
In the following discussions, for convenience, we shall not display the $t$ dependence in $P_{n}(x)$, $w(x)$, $h_{n}$, $\alpha_{n}$ and $\beta_{n}$ unless it is needed.
\begin{theorem}\label{lo}
Let $w_{0}(x)$ be a smooth weight function defined on $[a,b]$, and $w_{0}(a)=w_{0}(b)=0$. The monic orthogonal polynomials with respect to $w(x):=w_{0}(x)|x-t|^{\gamma}(A+B\theta(x-t)),\; t\in[a,b],\; \gamma>0$ satisfy the lowering operator equation
\be\label{lowering}
\left(\frac{d}{dz}+B_{n}(z)\right)P_{n}(z)=\beta_{n}A_{n}(z)P_{n-1}(z),
\ee
where
$$
A_{n}(z):=\frac{1}{h_{n}}\int_{a}^{b}\frac{\mathrm{v}_{0}'(z)-\mathrm{v}_{0}'(y)}{z-y}P_{n}^{2}(y)w(y)dy+a_{n}(z,t),
$$
$$
a_{n}(z,t):=\frac{\gamma}{h_{n}}\int_{a}^{b}\frac{P_{n}^{2}(y)}{(z-y)(y-t)}w(y)dy;
$$
$$
B_{n}(z):=\frac{1}{h_{n-1}}\int_{a}^{b}\frac{\mathrm{v}_{0}'(z)-\mathrm{v}_{0}'(y)}{z-y}P_{n}(y)P_{n-1}(y)w(y)dy+b_{n}(z,t),
$$
$$
b_{n}(z,t):=\frac{\gamma}{h_{n-1}}\int_{a}^{b}\frac{P_{n}(y)P_{n-1}(y)}{(z-y)(y-t)}w(y)dy
$$
and $\mathrm{v}_{0}(z)=-\ln w_{0}(z)$.
\end{theorem}

\begin{proof}
Since $P_{n}(z)$ is a polynomial of degree $n$, we have
$$
P_{n}'(z)=\sum_{k=0}^{n-1}C_{n,k}P_{k}(z)
$$
and the coefficient
\be\label{cnk}
C_{n,k}=\frac{1}{h_{k}}\int_{a}^{b}P_{n}'(y)P_{k}(y)w(y)dy.
\ee
After integration by parts and noting that $w(a)=w(b)=0$, we find
\bea\label{cnk1}
C_{n,k}&=&-\frac{1}{h_{k}}\int_{a}^{b}P_{n}(y)P_{k}(y)w'(y)dy\nonumber\\
&=&-\frac{1}{h_{k}}\int_{a}^{b}P_{n}(y)P_{k}(y)\bigg\{w_{0}'(y)|y-t|^{\gamma}(A+B\theta(y-t))dy+w_{0}(y)|y-t|^{\gamma}B\delta(y-t)\nonumber\\
&+&w_{0}(y)\bigg[\delta(y-t)((y-t)^{\gamma}-(t-y)^{\gamma})+\gamma\frac{|y-t|^{\gamma}}{y-t}\bigg](A+B\theta(y-t))\bigg\}dy\nonumber\\
&=&-\frac{1}{h_{k}}\int_{a}^{b}P_{n}(y)P_{k}(y)(-\mathrm{v}_{0}'(y))w(y)dy-\frac{\gamma}{h_{k}}\int_{a}^{b}P_{n}(y)P_{k}(y)\frac{w(y)}{y-t}dy\\
&=&-\frac{1}{h_{k}}\int_{a}^{b}P_{n}(y)P_{k}(y)(\mathrm{v}_{0}'(z)-\mathrm{v}_{0}'(y))w(y)dy-\frac{\gamma}{h_{k}}\int_{a}^{b}P_{n}(y)P_{k}(y)\frac{w(y)}{y-t}dy,
\nonumber
\eea
where we have used the formula \cite{ChenFeigin}
$$
\partial_{y}|y-t|^{\gamma}=\delta(y-t)((y-t)^{\gamma}-(t-y)^{\gamma})+\gamma\frac{|y-t|^{\gamma}}{y-t},
$$
which is obtained by writing $|y-t|^{\gamma}=(y-t)^{\gamma}\theta(y-t)+(t-y)^{\gamma}\theta(t-y)$.\\
It follows that
$$
P_{n}'(z)=-\int_{a}^{b}P_{n}(y)\sum_{k=0}^{n-1}\frac{P_{k}(z)P_{k}(y)}{h_{k}}(\mathrm{v}_{0}'(z)-\mathrm{v}_{0}'(y))w(y)dy
-\gamma\int_{a}^{b}P_{n}(y)\sum_{k=0}^{n-1}\frac{P_{k}(z)P_{k}(y)}{h_{k}}\frac{w(y)}{y-t}dy.
$$
By using the Christoffel-Darboux formula,
$$
\sum_{k=0}^{n-1}\frac{P_{k}(z)P_{k}(y)}{h_{k}}=\frac{P_{n}(z)P_{n-1}(y)-P_{n}(y)P_{n-1}(z)}{h_{n-1}(z-y)},
$$
we arrive at equation (\ref{lowering}).
\end{proof}

\begin{proposition}
We have the following two important identities:
\be\label{eq1}
\int_{a}^{b}P_{n}^2(y)\mathrm{v}_{0}'(y)w(y)dy=\gamma\int_{a}^{b}P_{n}^2(y)\frac{w(y)}{y-t}dy,
\ee
\be\label{eq2}
\frac{1}{h_{n-1}}\int_{a}^{b}P_{n}(y)P_{n-1}(y)\mathrm{v}_{0}'(y)w(y)dy=n+\frac{\gamma}{h_{n-1}}\int_{a}^{b}P_{n}(y)P_{n-1}(y)\frac{w(y)}{y-t}dy.
\ee
\end{proposition}
\begin{proof}
From (\ref{cnk}), we see that
$$
C_{n,n}=0,\;\;C_{n,n-1}=n.
$$
The equations (\ref{eq1}) and (\ref{eq2}) follow from (\ref{cnk1}) if we replace $k$ by $n$ and $n-1$, respectively.
\end{proof}

\begin{theorem}\label{s1s2}
The functions $A_{n}(z)$ and $B_{n}(z)$ satisfy the equations:
\be
B_{n+1}(z)+B_{n}(z)=(z-\alpha_{n})A_{n}(z)-\mathrm{v}_{0}'(z), \tag{$S_{1}$}
\ee
\be
1+(z-\alpha_{n})(B_{n+1}(z)-B_{n}(z))=\beta_{n+1}A_{n+1}(z)-\beta_{n}A_{n-1}(z). \tag{$S_{2}$}
\ee
\end{theorem}
\begin{proof}
From the definition of $B_n(z)$, we have
\bea\label{bpb}
B_{n+1}(z)+B_{n}(z)&=&\int_{a}^{b}P_{n}(y)\left(\frac{P_{n+1}(y)}{h_{n}}+\frac{P_{n-1}(y)}{h_{n-1}}\right)
\frac{\mathrm{v}_{0}'(z)-\mathrm{v}_{0}'(y)}{z-y}w(y)dy\nonumber\\
&+&\gamma\int_{a}^{b}P_{n}(y)\left(\frac{P_{n+1}(y)}{h_{n}}+\frac{P_{n-1}(y)}{h_{n-1}}\right)\frac{w(y)}{(z-y)(y-t)}dy.
\eea
It follows from the three-term recurrence relation (\ref{rr}) and (\ref{be}) that
$$
\frac{P_{n+1}(y)}{h_{n}}+\frac{P_{n-1}(y)}{h_{n-1}}=\frac{(y-\alpha_{n})P_{n}(y)}{h_{n}}.
$$
Substituting it into (\ref{bpb}) gives
$$
B_{n+1}(z)+B_{n}(z)=\frac{1}{h_{n}}\int_{a}^{b}(y-\alpha_{n})P_{n}^2(y)
\frac{\mathrm{v}_{0}'(z)-\mathrm{v}_{0}'(y)}{z-y}w(y)dy+\frac{\gamma}{h_{n}}\int_{a}^{b}(y-\alpha_{n})P_{n}^2(y)\frac{w(y)}{(z-y)(y-t)}dy.
$$
Using the definition of $A_n(z)$, it follows that
$$
B_{n+1}(z)+B_{n}(z)-(z-\alpha_{n})A_{n}(z)=-\frac{1}{h_{n}}\int_{a}^{b}\left(\mathrm{v}_{0}'(z)-\mathrm{v}_{0}'(y)\right)P_{n}^2(y)w(y)dy
-\frac{\gamma}{h_{n}}\int_{a}^{b}P_{n}^2(y)\frac{w(y)}{y-t}dy.
$$
From the orthogonality (\ref{ops}) and (\ref{eq1}), we find
$$
B_{n+1}(z)+B_{n}(z)-(z-\alpha_{n})A_{n}(z)=-\mathrm{v}_{0}'(z),
$$
which is just ($S_{1}$).\\
We now turn to prove ($S_{2}$). Similarly, by using the definition of $B_n(z)$, we have
\bea\label{bmb}
(z-\alpha_{n})(B_{n+1}(z)-B_{n}(z))
&=&\Bigg[\int_{a}^{b}(\mathrm{v}_{0}'(z)-\mathrm{v}_{0}'(y))P_{n}(y)
\left(\frac{P_{n+1}(y)}{h_{n}}-\frac{P_{n-1}(y)}{h_{n-1}}\right)w(y)dy\nonumber\\
&+&\gamma\int_{a}^{b}P_{n}(y)\left(\frac{P_{n+1}(y)}{h_{n}}-\frac{P_{n-1}(y)}{h_{n-1}}\right)\frac{w(y)}{y-t}dy\Bigg]\nonumber\\
&+&\Bigg[\int_{a}^{b}\frac{\mathrm{v}_{0}'(z)-\mathrm{v}_{0}'(y)}{z-y}(y-\alpha_{n})P_{n}(y)
\left(\frac{P_{n+1}(y)}{h_{n}}-\frac{P_{n-1}(y)}{h_{n-1}}\right)w(y)dy\nonumber\\
&+&\gamma\int_{a}^{b}(y-\alpha_{n})P_{n}(y)\left(\frac{P_{n+1}(y)}{h_{n}}-\frac{P_{n-1}(y)}{h_{n-1}}\right)\frac{w(y)}{(z-y)(y-t)}dy\Bigg],
\eea
where we write $z-\alpha_{n}=(z-y)+(y-\alpha_{n})$ to get two parts in (\ref{bmb}).\\
From (\ref{rr}), we have
$$
(y-\alpha_{n})P_{n}(y)=P_{n+1}(y)+\beta_{n}P_{n-1}(y).
$$
Substituting it into (\ref{bmb}), we find
\bea
(z-\alpha_{n})(B_{n+1}(z)-B_{n}(z))&=&\int_{a}^{b}(\mathrm{v}_{0}'(z)-\mathrm{v}_{0}'(y))P_{n}(y)
\left(\frac{P_{n+1}(y)}{h_{n}}-\frac{P_{n-1}(y)}{h_{n-1}}\right)w(y)dy\nonumber\\
&+&\gamma\int_{a}^{b}P_{n}(y)\left(\frac{P_{n+1}(y)}{h_{n}}-\frac{P_{n-1}(y)}{h_{n-1}}\right)\frac{w(y)}{y-t}dy\nonumber\\
&+&\int_{a}^{b}\frac{\mathrm{v}_{0}'(z)-\mathrm{v}_{0}'(y)}{z-y}\left(\frac{P_{n+1}^2(y)}{h_{n}}-\frac{\beta_{n}P_{n-1}^2(y)}{h_{n-1}}\right)w(y)dy\nonumber\\
&+&\gamma\int_{a}^{b}\left(\frac{P_{n+1}^2(y)}{h_{n}}-\frac{\beta_{n}P_{n-1}^2(y)}{h_{n-1}}\right)\frac{w(y)}{(z-y)(y-t)}dy.\nonumber
\eea
Using the definition of $A_n(z)$, it follows that
\bea
&&\beta_{n+1}A_{n+1}(z)-\beta_{n}A_{n-1}(z)-(z-\alpha_{n})(B_{n+1}(z)-B_{n}(z))\nonumber\\
&=&\int_{a}^{b}\left(\frac{P_{n}(y)P_{n-1}(y)}{h_{n-1}}-\frac{P_{n+1}(y)P_{n}(y)}{h_{n}}\right)(\mathrm{v}_{0}'(z)-\mathrm{v}_{0}'(y))w(y)dy\nonumber\\
&+&\gamma\int_{a}^{b}\left(\frac{P_{n}(y)P_{n-1}(y)}{h_{n-1}}-\frac{P_{n+1}(y)P_{n}(y)}{h_{n}}\right)\frac{w(y)}{y-t}dy\nonumber\\
&=&-n-[-(n+1)]\nonumber\\
&=&1,\nonumber
\eea
where we have used the orthogonality (\ref{ops}) and (\ref{eq2}). The proof is complete.
\end{proof}

The combination of ($S_{1}$) and ($S_{2}$) produces a sum rule.
\begin{theorem}\label{s2p}
$A_{n}(z)$, $B_{n}(z)$ and $\sum_{j=0}^{n-1}A_{j}(z)$ satisfy the equation
\be
B_{n}^{2}(z)+\mathrm{v}_{0}'(z)B_{n}(z)+\sum_{j=0}^{n-1}A_{j}(z)=\beta_{n}A_{n}(z)A_{n-1}(z). \tag{$S_{2}'$}
\ee
\end{theorem}
\begin{proof}
Multiplying ($S_{2}$) by $A_{n}(z)$ on both sides, we have
$$
A_{n}(z)+(z-\alpha_{n})A_{n}(z)(B_{n+1}(z)-B_{n}(z))=\beta_{n+1}A_{n+1}(z)A_{n}(z)-\beta_{n}A_{n}(z)A_{n-1}(z).
$$
Using ($S_{1}$), the above becomes
$$
A_{n}(z)+\left(B_{n+1}(z)+B_{n}(z)+\mathrm{v}_{0}'(z)\right)(B_{n+1}(z)-B_{n}(z))=\beta_{n+1}A_{n+1}(z)A_{n}(z)-\beta_{n}A_{n}(z)A_{n-1}(z),
$$
namely,
$$
A_{n}(z)+B_{n+1}^2(z)-B_{n}^2(z)+\mathrm{v}_{0}'(z)(B_{n+1}(z)-B_{n}(z))=\beta_{n+1}A_{n+1}(z)A_{n}(z)-\beta_{n}A_{n}(z)A_{n-1}(z).
$$
A telescopic sum gives the desired result.
\end{proof}
\begin{theorem}
The monic orthogonal polynomials $P_{n}(z)$ satisfy the raising operator equation
\be\label{raising}
\left(\frac{d}{dz}-B_{n}(z)-\mathrm{v}_{0}'(z)\right)P_{n-1}(z)=-A_{n-1}(z)P_{n}(z).
\ee
\end{theorem}
\begin{proof}
From (\ref{lowering}) we replace $n$ by $n-1$, it reads
\be\label{pnm1}
P_{n-1}'(z)=\beta_{n-1}A_{n-1}(z)P_{n-2}(z)-B_{n-1}(z)P_{n-1}(z).
\ee
The recurrence relation (\ref{rr}) gives
$$
\beta_{n-1}P_{n-2}(z)=(z-\alpha_{n-1})P_{n-1}(z)-P_{n}(z).
$$
Substituting it into (\ref{pnm1}), we obtain
\bea
P_{n-1}'(z)&=&\left[(z-\alpha_{n-1})A_{n-1}(z)-B_{n-1}(z)\right]P_{n-1}(z)-A_{n-1}(z)P_{n}(z)\nonumber\\
&=&(B_{n}(z)+\mathrm{v}_{0}'(z))P_{n-1}(z)-A_{n-1}(z)P_{n}(z).\nonumber
\eea
where we have made use of ($S_{1}$). This completes the proof.
\end{proof}

\begin{theorem}
The monic orthogonal polynomials $P_{n}(z)$ satisfy the second order differential equation
$$
P_{n}''(z)-\left(\mathrm{v}_{0}'(z)+\frac{A_{n}'(z)}{A_{n}(z)}\right)P_{n}'(z)+\left(B_{n}'(z)-B_{n}(z)\frac{A_{n}'(z)}{A_{n}(z)}
+\sum_{j=0}^{n-1}A_{j}(z)\right)P_{n}(z)=0.
$$
\end{theorem}
\begin{proof}
Solving for $P_{n-1}(z)$ from (\ref{lowering}) gives
$$
P_{n-1}(z)=\frac{P_{n}'(z)+B_{n}(z)P_{n}(z)}{\beta_{n}A_{n}(z)}.
$$
Substituting it into (\ref{raising}), we obtain
\bea
&&P_{n}''(z)-\left(\mathrm{v}_{0}'(z)+\frac{A_{n}'(z)}{A_{n}(z)}\right)P_{n}'(z)+\bigg(B_{n}'(z)-B_{n}(z)\frac{A_{n}'(z)}{A_{n}(z)}\nonumber\\
&+&\beta_{n}A_{n}(z)A_{n-1}(z)-B_{n}^{2}(z)-\mathrm{v}_{0}'(z)B_{n}(z)\bigg)P_{n}(z)=0.\nonumber
\eea
Using ($S_{2}'$), we obtain the desired result.
\end{proof}

\noindent $\mathbf{Remark\: 1.}$ In Theorem \ref{lo}, $a, b$ could be $-\infty$ or $\infty$.

\noindent $\mathbf{Remark\: 2.}$ The three identities ($S_{1}$), ($S_{2}$) and ($S_{2}'$) are valid for $z\in \mathbb{C}\cup\{\infty\}$.

\noindent $\mathbf{Remark\: 3.}$ The ladder operator approach has been widely applied to the study of orthogonal polynomials, Hankel determinants and random matrix theory; see \cite{ChenIts,ChenZhang,Dai,Filipuk,Min2018,MinLyuChen,VanAssche} for reference.

\section{Perturbed Laguerre Weight}
In this section, we apply the ladder operators and its supplementary conditions to our problem. For the problem at hand,
$$
w(x)=w_{0}(x)|x-t|^{\gamma}(A+B\theta(x-t)),\;\; x\in[0,\infty),\; t\in[0,\infty),
$$
where
$$
w_{0}(x)=x^{\alpha}\mathrm{e}^{-x},\;\;\mathrm{v}_{0}(x)=x-\alpha\ln x.
$$
It is easy to see that $w(0)=w(\infty)=0$ and
$$
\frac{\mathrm{v}_{0}'(z)-\mathrm{v}_{0}'(y)}{z-y}=\frac{\alpha}{zy}.
$$
From Theorem \ref{lo} we have
\be\label{anz}
A_{n}(z)=\frac{\alpha}{zh_{n}}\int_{0}^{\infty}\frac{P_{n}^{2}(y)w(y)}{y}dy+a_{n}(z,t),
\ee
$$
a_{n}(z,t)=\frac{\gamma}{h_{n}}\int_{0}^{\infty}\frac{P_{n}^{2}(y)}{(z-y)(y-t)}w(y)dy;
$$
\be\label{bnz}
B_{n}(z)=\frac{\alpha}{zh_{n-1}}\int_{0}^{\infty}\frac{P_{n}(y)P_{n-1}(y)w(y)}{y}dy+b_{n}(z,t),
\ee
$$
b_{n}(z,t)=\frac{\gamma}{h_{n-1}}\int_{0}^{\infty}\frac{P_{n}(y)P_{n-1}(y)}{(z-y)(y-t)}w(y)dy.
$$

\begin{theorem}
As $z\rightarrow\infty$, we have
\be\label{anz1}
A_{n}(z)=\frac{1}{z}+\frac{\gamma+t R_{n}(t)}{z^2}+\frac{\gamma\alpha_{n}+\gamma t+t^2 R_n(t)}{z^3}+O\left(\frac{1}{z^4}\right),
\ee
\be\label{bnz1}
B_{n}(z)=-\frac{n}{z}+\frac{t r_{n}(t)}{z^2}+\frac{\gamma\beta_{n}+t^2 r_n(t)}{z^3}+O\left(\frac{1}{z^4}\right),
\ee
where
$$
R_{n}(t):=\frac{\gamma}{h_{n}}\int_{0}^{\infty}\frac{P_{n}^{2}(y)}{y-t}w(y)dy,
$$
$$
r_{n}(t):=\frac{\gamma}{h_{n-1}}\int_{0}^{\infty}\frac{P_{n}(y)P_{n-1}(y)}{y-t}w(y)dy.
$$
\end{theorem}
\begin{proof}
Using integration by parts, we find
\bea
\frac{\alpha}{h_{n}}\int_{0}^{\infty}\frac{P_{n}^{2}(y)w(y)}{y}dy
&=&\frac{1}{h_{n}}\int_{0}^{\infty}P_{n}^{2}(y)\mathrm{e}^{-y}|y-t|^{\gamma}(A+B\theta(y-t))dy^{\alpha}\nonumber\\
&=&1-\frac{\gamma}{h_{n}}\int_{0}^{\infty}\frac{P_{n}^{2}(y)}{y-t}w(y)dy\nonumber\\
&=&1-R_{n}(t)\nonumber
\eea
and
$$
\frac{\alpha}{h_{n-1}}\int_{0}^{\infty}\frac{P_{n}(y)P_{n-1}(y)w(y)}{y}dy=-n-r_{n}(t).
$$
As $z\rightarrow\infty$,
$$
\frac{1}{z-y}=\frac{1}{z}\cdot\frac{1}{1-\frac{y}{z}}=\frac{1}{z}\left(1+\frac{y}{z}+\left(\frac{y}{z}\right)^2+O\left(\frac{1}{z^3}\right)\right)
=\frac{1}{z}+\frac{y}{z^2}+\frac{y^2}{z^3}+O\left(\frac{1}{z^4}\right).
$$
Then we have
$$
a_n(z,t)=\frac{\gamma}{zh_{n}}\int_{0}^{\infty}\frac{P_{n}^{2}(y)}{y-t}w(y)dy+\frac{\gamma}{z^2h_{n}}\int_{0}^{\infty}\frac{yP_{n}^{2}(y)}{y-t}w(y)dy
+\frac{\gamma}{z^3h_{n}}\int_{0}^{\infty}\frac{y^2P_{n}^{2}(y)}{y-t}w(y)dy+O\left(\frac{1}{z^4}\right),
$$
\bea
b_n(z,t)&=&\frac{\gamma}{zh_{n-1}}\int_{0}^{\infty}\frac{P_{n}(y)P_{n-1}(y)}{y-t}w(y)dy+\frac{\gamma}{z^2h_{n-1}}\int_{0}^{\infty}
\frac{yP_{n}(y)P_{n-1}(y)}{y-t}w(y)dy\nonumber\\
&+&\frac{\gamma}{z^3h_{n-1}}\int_{0}^{\infty}\frac{y^2P_{n}(y)P_{n-1}(y)}{y-t}w(y)dy+O\left(\frac{1}{z^4}\right).
\eea
By the definitions of $R_n(t)$ and $r_n(t)$, and using the orthogonality (\ref{ops}), also the recurrence relation (\ref{rr}), we obtain
$$
a_n(z,t)=\frac{R_{n}(t)}{z}+\frac{\gamma+t R_{n}(t)}{z^2}+\frac{\gamma\alpha_{n}+\gamma t+t^2 R_n(t)}{z^3}+O\left(\frac{1}{z^4}\right),
$$
$$
b_{n}(z,t)=\frac{r_n(t)}{z}+\frac{t r_{n}(t)}{z^2}+\frac{\gamma\beta_{n}+t^2 r_n(t)}{z^3}+O\left(\frac{1}{z^4}\right).
$$
According to (\ref{anz}) and (\ref{bnz}), the theorem is established.
\end{proof}
Substituting (\ref{anz1}) and (\ref{bnz1}) into ($S_{1}$), and comparing the coefficients of $\frac{1}{z}$ and $\frac{1}{z^2}$ on both sides respectively, we obtain the following two equations:
\be\label{s11}
\alpha_{n}=2n+1+\alpha+\gamma+t R_n(t),
\ee
\be\label{s12}
r_{n+1}(t)+r_{n}(t)=\gamma+(t-\alpha_{n})R_{n}(t).
\ee
Similarly, substituting (\ref{anz1}) and (\ref{bnz1}) into ($S_{2}$) gives rise to another two equations:
\be\label{s21}
\beta_{n+1}-\beta_{n}=tr_{n+1}(t)-tr_{n}(t)+\alpha_{n},
\ee
\be\label{s22}
(t-\alpha_{n})(r_{n+1}(t)-r_{n}(t))=\beta_{n+1}R_{n+1}(t)-\beta_{n}R_{n-1}(t).
\ee
Using (\ref{sum}), a telescopic sum of (\ref{s21}) gives
\be\label{s23}
\beta_{n}=tr_{n}(t)-\mathrm{p}(n,t).
\ee
Multiplying both sides of (\ref{s22}) by $R_{n}(t)$ and using (\ref{s12}), we have
$$
(r_{n+1}(t)+r_{n}(t)-\gamma)(r_{n+1}(t)-r_{n}(t))=\beta_{n+1}R_{n+1}(t)R_{n}(t)-\beta_{n}R_{n}(t)R_{n-1}(t)
$$
or
$$
r_{n+1}^2(t)-r_{n}^2(t)-\gamma(r_{n+1}(t)-r_{n}(t))=\beta_{n+1}R_{n+1}(t)R_{n}(t)-\beta_{n}R_{n}(t)R_{n-1}(t).
$$
A telescopic sum produces
\be\label{s24}
r_{n}^2(t)-\gamma r_{n}(t)=\beta_{n}R_{n}(t)R_{n-1}(t).
\ee

Substituting (\ref{anz1}) and (\ref{bnz1}) into ($S_{2}'$), noting that the coefficient of $\frac{1}{z}$ is 0 and comparing the coefficients of $\frac{1}{z^2}$ and $\frac{1}{z^3}$ on both sides respectively, we find the following two equations:
\be\label{s31}
n(n+\alpha+\gamma)+tr_n(t)+t\sum_{j=0}^{n-1}R_{j}(t)=\beta_{n},
\ee
\be\label{s32}
n\gamma t+(t^2-2nt-\alpha t)r_n(t)+\gamma\sum_{j=0}^{n-1}\alpha_{j}+t^2\sum_{j=0}^{n-1}R_{j}(t)=\beta_{n}\left(\gamma+tR_{n-1}(t)+tR_n(t)\right).
\ee
It follows from (\ref{sum}) and (\ref{s23}) that
$$
\sum_{j=0}^{n-1}\alpha_{j}=\beta_{n}-tr_{n}(t).
$$
Plugging it into (\ref{s32}) gives
\be\label{s33}
n\gamma+(t-2n-\alpha-\gamma)r_n(t)+t\sum_{j=0}^{n-1}R_{j}(t)=\beta_{n}R_{n-1}(t)+\beta_{n}R_{n}(t).
\ee
Eliminating $t\sum_{j=0}^{n-1}R_{j}(t)$ from (\ref{s31}) and (\ref{s33}), we obtain
\be\label{s34}
\beta_{n}R_{n-1}(t)+\beta_{n}R_{n}(t)=\beta_{n}-(2n+\alpha+\gamma)r_n(t)-n(n+\alpha).
\ee
In the end, we mention that the above identities obtained from ($S_{1}$), ($S_{2}$) and ($S_{2}'$) are very important for the derivation of the fifth Painlev\'{e} equation in next section.

\section{Painlev\'{e} V and Its $\sigma$-Form}
We start from taking a derivative with respect to $t$ in the following equation
$$
\int_{0}^{\infty}P_{n}^2(x,t)x^{\alpha}\mathrm{e}^{-x}|x-t|^{\gamma}(A+B\theta(x-t))dx=h_{n}(t),\;\;n=0,1,2,\ldots,
$$
which gives
$$
h_{n}'(t)=-\gamma\int_{0}^{\infty}P_{n}^2(x)\frac{w(x)}{x-t}dx.
$$
It follows that
\be\label{d1}
\frac{d}{dt}\ln h_{n}(t)=-R_n(t).
\ee
Using (\ref{be}) we have
$$
\frac{d}{dt}\ln \beta_{n}(t)=R_{n-1}(t)-R_n(t).
$$
That is,
\be\label{d11}
\beta_{n}'(t)=\beta_{n}R_{n-1}(t)-\beta_{n}R_n(t).
\ee
We define a quantity allied to the Hankel determinant,
\be\label{hnt}
H_{n}(t):=t\frac{d}{dt}\ln \mathcal{D}_{n}(t).
\ee
It is easy to see from (\ref{hankel}) and (\ref{d1}) that
\be\label{hn}
H_{n}(t)=-t\sum_{j=0}^{n-1}R_j(t).
\ee

On the other hand, taking a derivative with respect to $t$ in the equation
$$
\int_{0}^{\infty}P_{n}(x,t)P_{n-1}(x,t)x^{\alpha}\mathrm{e}^{-x}|x-t|^{\gamma}(A+B\theta(x-t))dx=0,\;\;n=0,1,2,\ldots,
$$
we obtain
\be\label{d2}
\frac{d}{dt}\mathrm{p}(n,t)=r_{n}(t).
\ee

\begin{theorem}
The auxiliary quantities $R_{n}(t)$ and $r_{n}(t)$ satisfy the following coupled Riccati equations:
\be\label{ri1}
t R_{n}'(t)=t R_{n}^2(t)+(2n+\alpha+\gamma-t)R_{n}(t)+2r_{n}(t)-\gamma,
\ee

\be\label{ri2}
t r_n'(t)=\frac{r_{n}^2(t)-\gamma r_{n}(t)}{R_{n}(t)}-\frac{r_{n}^2(t)-\gamma r_{n}(t)+[(2n+\alpha+\gamma)r_{n}(t)+n(n+\alpha)]R_{n}(t)}{1-R_{n}(t)}.
\ee
\end{theorem}

\begin{proof}
From (\ref{al}) and (\ref{d2}) we have
\be\label{alp}
\alpha_{n}'(t)=r_{n}(t)-r_{n+1}(t).
\ee
Eliminating $\alpha_{n}$ from (\ref{s11}) and (\ref{s12}) gives
\be\label{s13}
r_{n+1}(t)=\gamma+(t-2n-1-\alpha-\gamma-t R_n(t))R_{n}(t)-r_{n}(t).
\ee
Substituting (\ref{s11}) and (\ref{s13}) into (\ref{alp}), we obtain (\ref{ri1}).\\
From (\ref{s23}) and (\ref{d2}), we have
$$
\beta_{n}'(t)=t r_{n}'(t).
$$
Then (\ref{d11}) becomes
\be\label{minus}
t r_{n}'(t)=\beta_{n}R_{n-1}(t)-\beta_{n}R_n(t),
\ee
or
\be\label{d12}
t r_{n}'(t)=\frac{r_{n}^2(t)-\gamma r_{n}(t)}{R_n(t)}-\beta_{n}R_n(t),
\ee
where we have made use of (\ref{s24}).\\
From (\ref{s34}), we find the expression of $\beta_{n}$ in terms of $R_n(t)$ and $r_{n}(t)$ with the aid of (\ref{s24}),
\be\label{expre}
\beta_{n}=\frac{r_n^2(t)-\gamma r_n(t)}{R_n(t)(1-R_n(t))}+\frac{(2n+\alpha+\gamma)r_n(t)+n(n+\alpha)}{1-R_n(t)}.
\ee
Substituting it into (\ref{d12}), we arrive at (\ref{ri2}).
\end{proof}

\begin{theorem}
The quantity $R_n(t)$ satisfies a non-linear second order differential equation,
\bea\label{ode}
&&2t^2R_n(1-R_n)R_n''-t^2(1-2R_n)(R_n')^2+2tR_n(1-R_n)R_n'+2t^2R_n^5+t(4n+2+2\alpha+2\gamma-5t)R_n^4\nonumber\\
&-&4t(2n+1+\alpha+\gamma-t)R_n^3-[t^2-2(2n+1+\alpha+\gamma)t+\alpha^2-\gamma^2]R_n^2-2\gamma^2R_n+\gamma^2=0.
\eea
Let $S_{n}(t):=\frac{R_n(t)}{R_n(t)-1}$, then $S_{n}(t)$ satisfies the second order differential equation,
\be\label{pv}
S_{n}''=\frac{(3S_{n}-1)(S_{n}')^2}{2S_{n}(S_{n}-1)}-\frac{S_{n}'}{t}+\frac{(S_{n}-1)^2}{t^2}\left(\frac{\alpha^{2}S_{n}}{2}
-\frac{\gamma^2}{2S_{n}}\right)-\frac{(2n+1+\alpha+\gamma)S_{n}}{t}-\frac{S_{n}(S_{n}+1)}{2(S_{n}-1)}.
\ee
which is a particular Painlev\'{e} V, $P_{V}\left(\frac{\alpha^2}{2},-\frac{\gamma^2}{2},-(2n+1+\alpha+\gamma),-\frac{1}{2}\right)$, following the convention of \cite{Gromak}.
\end{theorem}

\begin{proof}
Solving for $r_n(t)$ from (\ref{ri1}) and substituting it into (\ref{ri2}), we obtain (\ref{ode}). After the linear fractional transformation $R_{n}(t)=\frac{S_n(t)}{S_n(t)-1}$ or $S_{n}(t)=\frac{R_n(t)}{R_n(t)-1}$, we arrive at (\ref{pv}).
\end{proof}

\noindent $\mathbf{Remark\: 4.}$ Solving for $R_n(t)$ from (\ref{ri2}) and substituting it into (\ref{ri1}), we can obtain the second order differential equation satisfied by $r_n(t)$. Since this equation is too complicated, we decide not to write it down. Usually the differential equation for $r_n(t)$ is related to the Chazy type equation; see \cite{Lyu2017,Min,Min2018} for reference.

\begin{theorem}
The Hankel determinant $\mathcal{D}_{n}(t)$ admits the following two alternative integral representations in terms of $R_n(t)$ and $S_n(t)$,
\bea
\ln\frac{\mathcal{D}_{n}(t)}{\mathcal{D}_{n}(0)}&=&\int_{0}^{t}\frac{1}{4s R_n(s)\left(R_n(s)-1\right)}\Big\{s^2(R_n'(s))^2-s^2 R_n^4(s)-2s(2n+\alpha+\gamma-s) R_n^3(s)\nonumber\\
&+&\left[2s(2n+\gamma)-(\alpha+\gamma-s)^2\right]R_n^2(s)+2\gamma(\alpha+\gamma-s)R_n(s)-\gamma^2\Big\}ds\nonumber\\
&=&\int_{0}^{t}\frac{1}{4s\:S_n(s)\left(S_n(s)-1\right)^2}\Big\{s^2(S_n'(s))^2-\alpha^2 S_n^4(s)-2\left[(2n+\alpha)s-\alpha^2+\alpha\gamma) \right]S_n^3(s)\nonumber\\
&-&\left[s^2-2s(2n+\alpha-\gamma)+\alpha^2-4\alpha\gamma+\gamma^2\right]S_n^2(s)-2\gamma(\alpha-\gamma-s)S_n(s)-\gamma^2\Big\}ds.\nonumber
\eea
\end{theorem}

\begin{proof}
Combining (\ref{s31}) and (\ref{hn}), we have
\be\label{exbe}
H_n(t)=n(n+\alpha+\gamma)+tr_n(t)-\beta_{n}.
\ee
Inserting (\ref{expre}) into the above and using (\ref{ri1}) to eliminate $r_n(t)$, we obtain the expression of $H_n(t)$ in terms of $R_n(t)$ and $R_n'(t)$,
\bea\label{hntex}
H_n(t)&=&\frac{1}{4 R_n(t)\left(R_n(t)-1\right) }\Big\{t^2 (R_n'(t))^2-t^2 R_n^4(t)-2 t(2n+\alpha +\gamma-t) R_n^3(t)\nonumber\\ &+&\left[2t(2n+\gamma)-(\alpha+\gamma-t)^2\right]R_n^2(t)+2 \gamma   (\alpha +\gamma -t)R_n(t)-\gamma ^2\Big\}.
\eea
Since $S_{n}(t)=\frac{R_n(t)}{R_n(t)-1}$, we also have
\bea
H_n(t)&=&\frac{1}{4S_n(t)\left(S_n(t)-1\right)^2}\Big\{t^2 (S_n'(t))^2-\alpha^2 S_n^4(t)-2 \left[(2n+\alpha)t-\alpha^2 +\alpha\gamma) \right]S_n^3(t)\nonumber\\
&-&\left[t^2-2 t(2n+\alpha-\gamma)+\alpha ^2-4 \alpha  \gamma +\gamma ^2\right]S_n^2(t)-2 \gamma  (\alpha-\gamma-t)S_n(t)-\gamma ^2\Big\}.\nonumber
\eea
In view of $H_n(t)=t\frac{d}{dt}\ln \mathcal{D}_{n}(t)$, the theorem is established.
\end{proof}
\begin{theorem}\label{thm1}
The quantity $H_{n}(t)$ satisfies a non-linear second order differential equation
\be\label{hd}
(t H_n'')^2=\left[n\gamma-H_n-(2n+\alpha+\gamma-t)H_n'\right]^2-4\left[n(n+\alpha+\gamma)-H_n+tH_n'\right]\left[(H_n')^2-\gamma H_n'\right]
\ee
with the initial conditions $H_n(0)=0$, $H_n'(0)=-\frac{n\gamma}{\alpha+\gamma}$, and also satisfies a non-linear second order difference equation
\bea\label{hnd}
&&\Big\{n\gamma t-t H_n-\left[n(n+\alpha+\gamma)-H_n\right](H_{n-1}-H_{n+1})\Big\}\Big\{(t-n-\alpha-\gamma)\gamma t-t H_n\nonumber\\
&-&\left[n(n+\alpha+\gamma)+\gamma t-H_n\right](H_{n-1}-H_{n+1})\Big\}=\Big[n\gamma t+n(n+\alpha+\gamma)(2n+\alpha+\gamma-t)-(2n+\alpha+\gamma)H_n\Big]\nonumber\\
&\cdot&(2n+\alpha+\gamma-t+H_{n-1}-H_{n+1})(H_n-H_{n+1})(H_{n-1}-H_{n}).
\eea
\end{theorem}

\begin{proof}
From (\ref{s23}) and (\ref{exbe}), we have
$$
\mathrm{p}(n,t)=-n(n+\alpha+\gamma)+H_n(t).
$$
Taking a derivative with respect to $t$ on both sides and noting (\ref{d2}), we find
\be\label{rn}
r_n(t)=H_n'(t).
\ee
It follows from (\ref{exbe}) that
\be\label{beta}
\beta_{n}=n(n+\alpha+\gamma)-H_n(t)+tH_n'(t).
\ee
From (\ref{s34}) and (\ref{minus}), it is easy to get
$$
2\beta_{n}R_n(t)=\beta_{n}-(2n+\alpha+\gamma)r_n(t)-n(n+\alpha)-tr_n'(t),
$$
$$
2\beta_{n}R_{n-1}(t)=\beta_{n}-(2n+\alpha+\gamma)r_n(t)-n(n+\alpha)+tr_n'(t).
$$
The product of the above two equations gives
\be\label{rnbe}
4\beta_{n}(r_{n}^2(t)-\gamma r_{n}(t))=\left[\beta_{n}-(2n+\alpha+\gamma)r_n(t)-n(n+\alpha)\right]^2-(tr_n'(t))^2,
\ee
where we have made use of (\ref{s24}).\\
Substituting (\ref{rn}) and (\ref{beta}) into (\ref{rnbe}), we obtain (\ref{hd}). The initial conditions come from (\ref{hn}) and the fact that $R_n(0)=\frac{\gamma}{\alpha+\gamma}$.

We now turn to prove the difference equation satisfied by $H_n(t)$. From (\ref{hn}) we have
\be\label{rh1}
tR_n(t)=H_n(t)-H_{n+1}(t),
\ee
\be\label{rh2}
tR_{n-1}(t)=H_{n-1}(t)-H_{n}(t).
\ee
Multiplying both sides of (\ref{s34}) by $t$ and substituting (\ref{exbe}), (\ref{rh1}) and (\ref{rh2}) into it, we obtain the expression of $r_n(t)$ in terms of $H_n(t)$,
\be\label{rnt}
t r_n(t)=\frac{n\gamma t-t H_n(t)-\left[n(n+\alpha+\gamma)-H_n(t)\right](H_{n-1}(t)-H_{n+1}(t))}{2n+\alpha+\gamma-t+H_{n-1}(t)-H_{n+1}(t)}.
\ee
It follows from (\ref{exbe}) that
\be\label{betan}
\beta_{n}=\frac{n\gamma t+n(n+\alpha+\gamma)(2n+\alpha+\gamma-t)-(2n+\alpha+\gamma)H_n(t)}{2n+\alpha+\gamma-t+H_{n-1}(t)-H_{n+1}(t)}.
\ee
Finally, multiplying (\ref{s24}) by $t^2$ on both sides and substituting (\ref{rh1}), (\ref{rh2}), (\ref{rnt}) and (\ref{betan}) into it, we obtain (\ref{hnd}). The proof is complete.
\end{proof}

From the above theorem, we readily have the following results, which connect our problem with the Painlev\'{e} equations.
\begin{theorem}\label{thm2}
Let $\sigma_{n}(t):=H_{n}(t)-n\gamma$, then $\sigma_{n}(t)$ satisfies the Jimbo-Miwa-Okamoto $\sigma$-form of the Painlev\'{e} V \cite{Jimbo1981},
$$
(t\sigma_{n}'')^2=\left[\sigma_{n}-t\sigma_{n}'+2(\sigma_{n}')^2+(\nu_{0}+\nu_{1}+\nu_{2}+\nu_{3})\sigma_{n}'\right]^2
-4(\nu_{0}+\sigma_{n}')(\nu_{1}+\sigma_{n}')(\nu_{2}+\sigma_{n}')(\nu_{3}+\sigma_{n}'),
$$
where $\nu_{0}=0,\;\nu_{1}=n,\;\nu_{2}=n+\alpha,\;\nu_{3}=-\gamma$, with the initial conditions $\sigma_n(0)=-n\gamma,\; \sigma_n'(0)=-\frac{n\gamma}{\alpha+\gamma}$.\\
The quantity $\sigma_{n}(t)$ also satisfies a non-linear second order difference equation
\bea
&&\big[t \sigma_n-(n^2+n\alpha-\sigma_n)(2\gamma-\sigma_{n-1}+\sigma_{n+1})\big]\big[(2n+\alpha+\gamma-t)\gamma t+t\sigma_n\nonumber\\
&-&(n^2+n\alpha+\gamma t-\sigma_n)(2\gamma-\sigma_{n-1}+\sigma_{n+1})\big]=\big[n(n+\alpha)(2n+\alpha+\gamma-t)-(2n+\alpha+\gamma)\sigma_n\big]\nonumber\\
&\cdot&(2n+\alpha-\gamma-t+\sigma_{n-1}-\sigma_{n+1})(\gamma-\sigma_n+\sigma_{n+1})(\gamma-\sigma_{n-1}+\sigma_{n}),\nonumber
\eea
which is the discrete $\sigma$-form of the Painlev\'{e} V.
\end{theorem}
\noindent $\mathbf{Remark\: 5.}$ If $\gamma=0$, then the results in Theorem \ref{thm1} or Theorem \ref{thm2} are coincident with Theorem 8 in Basor and Chen \cite{Basor2009}.

In the end of this section, we show the relation of our Hankel determinant with the Toda molecule equation in the following theorem.
\begin{theorem}
The Hankel determinant $\mathcal{D}_n(t)$ satisfies the following differential-difference equation,
\be\label{td1}
t^2\frac{d^2}{dt^2}\ln\mathcal{D}_n(t)=-n(n+\alpha+\gamma)+\frac{\mathcal{D}_{n+1}(t)\mathcal{D}_{n-1}(t)}{\mathcal{D}_n^2(t)}.
\ee
Furthermore, let $\tilde{\mathcal{D}}_n(t):=t^{-n(n+\alpha+\gamma)}\mathcal{D}_n(t)$, then $\tilde{\mathcal{D}}_n(t)$ satisfies the Toda molecule equation \cite{Sogo}
\be\label{td2}
\frac{d^2}{dt^2}\ln\tilde{\mathcal{D}}_n(t)=\frac{\tilde{\mathcal{D}}_{n+1}(t)\tilde{\mathcal{D}}_{n-1}(t)}{\tilde{\mathcal{D}}_n^2(t)}.
\ee
\end{theorem}

\begin{proof}
From (\ref{be}) and (\ref{hankel}), we have
\be\label{b1}
\beta_{n}=\frac{\mathcal{D}_{n+1}(t)\mathcal{D}_{n-1}(t)}{\mathcal{D}_n^2(t)}.
\ee
On the other hand, from (\ref{beta}) and (\ref{hnt}) we find
\be\label{b2}
\beta_n=n(n+\alpha+\gamma)+t^2\frac{d^2}{dt^2}\ln\mathcal{D}_n(t).
\ee
The combination of (\ref{b1}) and (\ref{b2}) gives (\ref{td1}). The equation (\ref{td2}) follows from the transformation $\tilde{\mathcal{D}}_n(t)=t^{-n(n+\alpha+\gamma)}\mathcal{D}_n(t)$. This completes the proof.
\end{proof}

\noindent $\mathbf{Remark\: 6.}$ The Hankel determinant $\mathcal{D}_n(t)$ is related to the $\tau$-function of the Painlev\'{e} V \cite{Okamoto}. See also \cite{Forrester2010,Forrester2002} on the discussion of the $\tau$-functions and the Painlev\'{e} equations.

\section{Asymptotics}
In the limit of large $n$, the eigenvalues (particles) of the Hermitian matrices from a unitary ensemble can be approximated as a continuous fluid with a density $\sigma(x)$ supported in $J$ (a subset of $\mathbb{R}$). When the potential $\mathrm{v}(x):=-\ln w(x)$ is convex and $\mathrm{v}''(x)>0$ in a set of positive measure, $\sigma(x)$ is supported in a single interval $(a,b)$. See \cite{ChenIsmail,Chen1998} for detail.

The equilibrium density $\sigma(x)$ is found to satisfy the following singular integral equation,
$$
\mathrm{v}'(x)-2P\int_{a}^{b}\frac{\sigma(y)}{x-y}dy=0,
$$
where $P$ denotes the principal value integral.
\\
The solution subject to the boundary condition $\sigma(a)=\sigma(b)=0$ reads,
$$
\sigma(x)=\frac{\sqrt{(b-x)(x-a)}}{2\pi^2}P\int_{a}^{b}\frac{\mathrm{v}'(y)}{(y-x)\sqrt{(b-y)(y-a)}}dy
$$
with two supplementary conditions
\be\label{sup1}
\int_{a}^{b}\frac{\mathrm{v}'(x)}{\sqrt{(b-x)(x-a)}}dx=0,
\ee
\be\label{sup2}
\int_{a}^{b}\frac{x\:\mathrm{v}'(x)}{\sqrt{(b-x)(x-a)}}dx=2\pi n.
\ee

For our problem,
$$
w(x,t)=x^{\alpha}\mathrm{e}^{-x}|x-t|^{\gamma}(A+B\theta(x-t))=:\mathrm{e}^{-\mathrm{v}(x)},
$$
where
\be\label{vx}
\mathrm{v}(x)=x-\alpha\ln x-\gamma\ln|x-t|-\theta(x-t)\ln(A+B)-\theta(t-x)\ln A.
\ee
Substituting (\ref{vx}) into (\ref{sup1}) and (\ref{sup2}) and noting that $\frac{d}{dx}\theta(x-t)=\delta(x-t)$, we obtain two equations for the endpoints $a$ and $b$ ($0<a<b$):
\be\label{equ1}
1-\frac{\alpha}{\sqrt{ab}}+\frac{c}{\sqrt{(b-t)(t-a)}}=0,
\ee
\be\label{equ2}
\frac{a+b}{2}-\alpha-\gamma+\frac{c\:t}{\sqrt{(b-t)(t-a)}}=2n,
\ee
where $c:=\frac{1}{\pi}\ln\frac{A}{A+B}$ and we have used the following formulas \cite{ChenMcKay2012,Gradshteyn},
$$
\int_{a}^{b}\frac{1}{\sqrt{(b-x)(x-a)}}dx=\pi,
$$
$$
\int_{a}^{b}\frac{x}{\sqrt{(b-x)(x-a)}}dx=\frac{a+b}{2}\pi,
$$
$$
\int_{a}^{b}\frac{1}{x\sqrt{(b-x)(x-a)}}dx=\frac{\pi}{\sqrt{ab}},\quad (0<a<b),
$$
$$
P\int_{a}^{b}\frac{1}{(x-t)\sqrt{(b-x)(x-a)}}dx=0.
$$

Since $\alpha_{n}\sim\frac{a+b}{2}$ as $n\rightarrow\infty$ \cite{ChenIsmail}, we denote $\tilde{\alpha}_{n}:=\frac{a+b}{2}$. From (\ref{equ1}) and (\ref{equ2}) we obtain a quintic equation satisfied by $\tilde{\alpha}_{n}$,
\be\label{alpha}
(\tilde{\alpha}_{n}-2n-\alpha-\gamma)^2\left[(2\tilde{\alpha}_{n}-t)(\tilde{\alpha}_{n}-t-2n-\alpha-\gamma)^2-\alpha^2t\right]
-c^2t(\tilde{\alpha}_{n}-t-2n-\alpha-\gamma)^2=0.
\ee
In view of the relation (\ref{s11}), letting $\tilde{\alpha}_{n}=2n+\alpha+\gamma+t\tilde{R}_{n}(t)$, we have $\tilde{R}_{n}(t)\sim R_{n}(t)$ as $n\rightarrow\infty$. It follows from (\ref{alpha}) that
\bea\label{ale}
&&2t^2\tilde{R}_{n}^5(t)-t(5t-4n-2\alpha-2\gamma)\tilde{R}_{n}^4(t)+4t(t-2n-\alpha-\gamma)\tilde{R}_{n}^3(t)\nonumber\\
&-&(t^2-4nt-2\alpha t-2\gamma t+\alpha^2+c^2)\tilde{R}_{n}^2(t)+2c^2 \tilde{R}_{n}(t)-c^2=0.
\eea

In the end, we consider the case when $t$ approaches the soft edge, i.e., $n\rightarrow\infty, t=4n+2^{\frac{4}{3}}n^{\frac{1}{3}}s$ and $s$ is fixed. The asymptotic behavior of $R_n(t)$, $r_n(t)$ and $H_n(t)$ is obtained in the following theorem.
\begin{theorem}
Assume that $n\rightarrow\infty, t=4n+2^{\frac{4}{3}}n^{\frac{1}{3}}s$ and $s$ is fixed. Then the large $n$ asymptotics of $R_n(t), r_n(t)$ and $H_n(t)$ are given by
$$
R_n(t)=n^{-\frac{2}{3}}u(s)+n^{-1}v(s)+O(n^{-\frac{4}{3}}),
$$
$$
r_n(t)=n^{\frac{1}{3}}u(s)+2^{-\frac{1}{3}}u'(s)+v(s)+\frac{\gamma}{2}+O(n^{-\frac{1}{3}}),
$$
and
$$
H_n(t)=2 \gamma  n-\frac{2^{\frac{4}{3}} (u'(s))^2+16 u^3(s)-2^{\frac{10}{3}}s\: u^2(s)-\gamma ^2}{4 u(s)}n^{\frac{2}{3}}+O(n^{\frac{1}{3}}),
$$
respectively. Here $u(s)$ and $v(s)$ satisfy the second order differential equations (\ref{us}) and (\ref{vs}), and the large $s$ behavior is given by (\ref{us1}) and (\ref{vs1}). In addition, $\tilde{u}(s):=-2^{\frac{2}{3}}u(s)$ satisfies the Painlev\'{e} XXXIV equation \cite{Ince}
\be\label{p34}
\tilde{u}''(s)=\frac{(\tilde{u}'(s))^2}{2\tilde{u}(s)}+4\tilde{u}^2(s)+2s\tilde{u}(s)-\frac{\gamma^2}{2\tilde{u}(s)}.
\ee
\end{theorem}
\begin{proof}
Let
$$
\hat{R}_n(s):=R_n(4n+2^{\frac{4}{3}}n^{\frac{1}{3}}s).
$$
After change of variable, equation (\ref{ode}) becomes
\bea\label{leq}
&&2\left(2^{\frac{2}{3}}n^{\frac{2}{3}}+s\right)^2\left(1-\hat{R}_n(s)\right) \hat{R}_n(s)  \hat{R}_n''(s)-\left(2^{\frac{2}{3}}n^{\frac{2}{3}}+s\right)^2\left(1-2 \hat{R}_n(s)\right) \left(\hat{R}_n'(s)\right)^2\nonumber\\
&+&\left(2^{\frac{5}{3}}n^{\frac{2}{3}}+2s\right)(1-\hat{R}_n(s)) \hat{R}_n(s)  \hat{R}_n'(s)+2  \left(4n+2^{\frac{4}{3}}n^{\frac{1}{3}}s\right)^2\hat{R}_n^5(s)\nonumber\\
&+&\left(4n+2^{\frac{4}{3}}n^{\frac{1}{3}}s\right) \left(2 \alpha +2 \gamma+2-16n -5\times2^{\frac{4}{3}}n^{\frac{1}{3}} s\right)\hat{R}_n^4(s)\nonumber\\
&-&4  \left(4n+2^{\frac{4}{3}}n^{\frac{1}{3}}s\right) \left(\alpha +\gamma-2 n+1 -2^{\frac{4}{3}}n^{\frac{1}{3}} s\right)\hat{R}_n^3(s)\nonumber\\
&+&\left[\alpha ^2-\gamma ^2-2 (\alpha +\gamma +2 n+1)\left(4n+2^{\frac{4}{3}}n^{\frac{1}{3}}s\right) +\left(4n+2^{\frac{4}{3}}n^{\frac{1}{3}}s\right)^2\right]\hat{R}_n^2(s)\nonumber\\
&-&2 \gamma ^2 \hat{R}_n(s)+\gamma ^2=0.
\eea
We suppose
\be\label{rex}
\hat{R}_n(s)=n^{-\frac{2}{3}}u(s)+n^{-1}v(s)+O(n^{-\frac{4}{3}}),
\ee
which is obtained by observing from the real solution of the algebraic equation (\ref{ale}) after changing variable $t$ to $s$.

Substituting (\ref{rex}) into (\ref{leq}), we obtain
\bea
&&u(s) u''(s)-2^{-1} (u'(s))^2+2^{\frac{8}{3}} u^3(s)-2 s\: u^2(s)+2^{-\frac{7}{3}}\gamma ^2+n^{-\frac{1}{3}}\Big[u(s) v''(s)-u'(s) v'(s)+ u''(s)v(s)\nonumber\\
&+&12\times2^{\frac{2}{3}} u^2(s) v(s)-4 s\: u(s) v(s)+2^{\frac{2}{3}} (\alpha +\gamma +1) u^2(s)\Big]+O(n^{-\frac{2}{3}})=0.\nonumber
\eea
It follows that $u(s)$ and $v(s)$ satisfy the following second order differential equations
\be\label{us}
u(s) u''(s)-2^{-1} (u'(s))^2+2^{\frac{8}{3}} u^3(s)-2 s\: u^2(s)+2^{-\frac{7}{3}}\gamma ^2=0,
\ee
\be\label{vs}
u(s) v''(s)-u'(s) v'(s)+ u''(s)v(s)+12\times2^{\frac{2}{3}} u^2(s) v(s)-4 s\: u(s) v(s)+2^{\frac{2}{3}} (\alpha +\gamma +1) u^2(s)=0.
\ee
From (\ref{us}), we obtain the large $s$ asymptotic of $u(s)$. As $s\rightarrow\infty$,
\be\label{us1}
u(s)=\frac{s}{2^{5/3}}+\frac{1-4 \gamma ^2}{8\times2^{2/3} s^2}-\frac{16 \gamma ^4-40 \gamma ^2+9}{16\times2^{2/3} s^5}-\frac{7 \left(64 \gamma ^6-496 \gamma ^4+876 \gamma ^2-189\right)}{128\times2^{2/3} s^8}+O\left(\frac{1}{s^{11}}\right).
\ee
Substituting (\ref{us1}) into (\ref{vs}), we find the large $s$ asymptotic of $v(s)$. As $s\rightarrow\infty$,
\bea\label{vs1}
v(s)&=&-\frac{1}{4} (\alpha +\gamma +1)-\frac{\left(4 \gamma ^2-1\right) (\alpha +\gamma +1)}{8 s^3}-\frac{5 \left(16 \gamma ^4-40 \gamma ^2+9\right) (\alpha +\gamma +1)}{32 s^6}\nonumber\\
&-&\frac{7 \left(64 \gamma ^6-496 \gamma ^4+876 \gamma ^2-189\right) (\alpha +\gamma +1)}{32 s^9}+O\left(\frac{1}{s^{12}}\right).
\eea
Let
$$
\hat{r}_n(s):=r_n(4n+2^{\frac{4}{3}}n^{\frac{1}{3}}s).
$$
From (\ref{ri1}), we obtain the expression of $\hat{r}_n(s)$ in terms of $\hat{R}_n(s)$,
\be\label{rer}
2\hat{r}_n(s)=(2^{\frac{2}{3}}n^{\frac{2}{3}}+s)\hat{R}_n'(s)-(4n+2^{\frac{4}{3}}n^{\frac{1}{3}}s)\hat{R}_n^2(s)
-(\alpha+\gamma-2n-2^{\frac{4}{3}}n^{\frac{1}{3}}s)\hat{R}_n(s)+\gamma.
\ee
Substituting (\ref{rex}) into (\ref{rer}) gives
$$
\hat{r}_n(s)=n^{\frac{1}{3}}u(s)+2^{-\frac{1}{3}}u'(s)+v(s)+\frac{\gamma}{2}+O(n^{-\frac{1}{3}}).
$$
Similarly, we have the expression of $\hat{H}_n(s):=H_n(4n+2^{\frac{4}{3}}n^{\frac{1}{3}}s)$ in terms of $\hat{R}_n(s)$ from (\ref{hntex}). Using (\ref{rex}), we obtain
$$
\hat{H}_n(s)=2 \gamma  n-\frac{2^{\frac{4}{3}} (u'(s))^2+16 u^3(s)-2^{\frac{10}{3}}s\: u^2(s)-\gamma ^2}{4 u(s)}n^{\frac{2}{3}}+O(n^{\frac{1}{3}}).
$$
In the end, letting $\tilde{u}(s):=-2^{\frac{2}{3}}u(s)$, then it follows from (\ref{us}) that $\tilde{u}(s)$ satisfies the Painlev\'{e} XXXIV equation (\ref{p34}).
\end{proof}

\noindent $\mathbf{Remark\: 7.}$ The Painlev\'{e} XXXIV equation (\ref{p34}) also appeared in the study of the critical edge behavior in quite general unitary random matrix ensembles by Its, Kuijlaars and \"{O}stensson \cite{Its2008, Its2009}, and the study of perturbed Gaussian unitary ensemble with a Fisher-Hartwig singularity by Wu, Xu and Zhao \cite{Wu}.

\section*{Acknowledgments}
Chao Min was supported by the Scientific Research Funds of Huaqiao University under grant number 600005-Z17Y0054.
Yang Chen was supported by the Macau Science and Technology Development Fund under grant numbers FDCT 130/2014/A3, FDCT 023/2017/A1 and by the University of Macau under grant numbers MYRG 2014-00011-FST, MYRG 2014-00004-FST.

\end{document}